\documentclass[hidelinks,11pt,english]{article}
\usepackage[T1]{fontenc}
\usepackage[latin9]{inputenc}
\usepackage{lmodern}
\usepackage{amsmath}
\usepackage{amsthm}
\usepackage{amssymb}
\usepackage{stmaryrd}
\usepackage{graphicx}
\usepackage{xcolor}
\usepackage{mathrsfs}
\usepackage{mathtools}
\usepackage{braket}
\usepackage{upgreek}
\usepackage{babel}
\usepackage{verbatim}
\usepackage{enumerate}
\usepackage{thmtools}
\usepackage{thm-restate}
\usepackage{cite}
\usepackage{authblk}
\usepackage{hyperref}

\declaretheorem[name=Lemma]{lem}
\declaretheorem[name=Observation]{obs}

\makeatletter
\numberwithin{equation}{section}

\newtheorem{theorem}{Theorem}

\newtheorem{corollary}{Corollary}
\theoremstyle{definition}
\newtheorem{definition}{Definition}

\DeclareMathOperator{\Tr}{Tr}

\DeclareMathOperator{\I}{I}
\DeclareMathOperator{\relint}{relint}

\newcommand{\C}{\mathbb{C}}
\newcommand{\R}{\mathbb{R}}
\newcommand{\cH}{\mathcal{H}}
\newcommand{\PH}{\mathcal{P}(\mathcal{H})}
\newcommand{\EH}{\mathcal{E}(\mathcal{H})}
\newcommand{\ES}{E(\mathcal{S})}
\newcommand{\E}{\mathcal{E}}
\newcommand{\cS}{\mathcal{S}}

\newcommand{\ef}{{\boldsymbol{e}}}

\newcommand{\ff}{{\boldsymbol{f}}}
\newcommand{\bo}{{\boldsymbol{\omega}}}
\newcommand{\bu}{{\boldsymbol{u}}}
\newcommand{\bx}{\boldsymbol{x}}
\newcommand{\by}{\boldsymbol{y}}
\newcommand{\bO}{\boldsymbol{0}}
\newcommand{\bh}{{\boldsymbol{h}}}

\newcommand{\bm}{{\boldsymbol{m}}}

\begin{document}
\title{Intermediate determinism in general probabilistic theories}

\author{Victoria J Wright}
\date{{\small {\it ICFO-Institut de Ciencies Fotoniques, The Barcelona Institute of Science and Technology, 08860 Castelldefels, Spain} \\\href{mailto:victoria.wright@icfo.eu}{victoria.wright@icfo.eu}}\\\vspace{0.5cm}\today}

\maketitle
\begin{abstract}

Quantum theory is indeterministic, but not completely so. When a system is in a pure state there are properties it possesses with certainty, known as actual properties. The actual properties of a quantum system (in a pure state) fully determine
the probability of finding the system to have any other property. We call this feature intermediate determinism. In dimensions of at least three, the intermediate determinism of quantum theory is guaranteed by the structure of its lattice of properties. This observation follows from Gleason's theorem, which is why it fails to hold in dimension two. In this work we extend the idea of intermediate determinism from properties to measurements. Under this extension intermediate determinism follows from the structure of quantum effects for separable Hilbert spaces of any dimension, including dimension two. Then, we find necessary and sufficient conditions for a general probabilistic theory to obey intermediate determinism. We show that, although related, both the no-restriction hypothesis and a Gleason-type theorem are neither necessary nor sufficient for intermediate determinism. 
\end{abstract}

\section{Introduction}
It is well-known that in quantum theory properties of physical systems cannot be predicted with certainty. On the contrary, even with perfect knowledge of the state of the system one may only deduce its \emph{propensity to actualise}\footnote{We follow the usage of propensity in Ref.~\cite{gisin1991propensities} to mean the probability of a non-necessary (and non-predetermined) event. Accordingly, the propensity of a system to actualise a property would be the propensity of an outcome of an ideal measurement indicating that the system has that property immediately after the measurement.} a given property. Although quantum theory contains this intrinsic uncertainty it is not completely devoid of deterministic measurement events. Given an eigenstate of an observable $A$ with eigenvalue $a$, a system in this eigenstate will be deterministically found to have the property $A=a$ upon the measurement of $A$. We say that this property is an \emph{actual} property of the system in this state. 

Gisin~\cite{gisin1991propensities} defines an intermediate level between complete determinism and complete randomness wherein the \emph{actual} properties of the system fully determine its propensity to take other properties. In other words, the (pure) state of the system is entirely determined by its actual properties. Gisin proposes this intermediate determinism along with an axiom saying that every property should be an actual property of some state as a means to single out classical and quantum physics from other candidate theories. 




It remains an open question whether there exist theories satisfying Gisin's axioms in addition to those described by quantum and classical physics. However, even upon a negative resolution to this question the approach would not allow for a \emph{derivation} of quantum theory in its entirety; the proposed axioms rule out two-level quantum systems. This exclusion follows from the absence of Gleason's theorem for two dimensional complex Hilbert spaces. Although two-level quantum systems satisfy intermediate determinism, this fact does not follow from the system's property lattice. The lattice allows in theory for so many states that none of them have a unique set of actual properties.

Two-level systems do, however, admit a \emph{Gleason-type} theorem when one considers generalised observables given by positive-operator-valued measures (POVMs)~\cite{Busch2003,Caves2004,Wright2018}. Additionally, Gleason-type theorems have been proven for general probabilistic theories (GPTs)~\cite{GPTGTT}, a broad class of theories derived from operational assumptions. In this work we investigate intermediate determinism in theories that do admit Gleason-type results. 

Although originally defined in terms of properties, intermediate determinism can more generally be applied to measurements. The analogous statement is that every pure state of the system is uniquely identified by its actual \emph{effects}. Effects are elements of a physical theory assigned to all the possible outcomes of measuring all the possible observables to represent the relationships between these outcomes. For example, consider the measurement of a pair of observables $A$ and $B$ which have as possible outcomes $a$ and $b$, respectively, such that the probability of observing $a$ is equal to that of $b$ for any state of the system. This relationship is represented by assigning the same effect to both outcomes.

First, we consider this generalised notion of intermediate determinism in quantum theory. Under the generalisation, quantum theory continues to exhibit intermediate determinism, however, this fact now follows from the structure of quantum effects for all separable Hilbert spaces, including in dimension two. Therefore, analogous axioms for effects as those proposed by Gisin for properties would no longer rule out two-level quantum systems. 


Second, we consider intermediate determinism in GPTs, which becomes possible after the extension to effects. We say that the principle is satisfied when the extremal points of the state space have a unique set of actual effects. Unlike quantum theory, not all GPTs have this property. As our main result, Theorem~\ref{thm:main}, we provide necessary and sufficient conditions for a GPT to obey intermediate determinism. Although related, satisfaction of the no-restriction hypothesis~\cite{ChiribellaPhysRevA.81.062348,janottanorestriction} and the existence of a Gleason-type theorem~\cite{GPTGTT} are both found to be neither necessary nor sufficient for intermediate determinism. 

The aim of this work is, firstly, to better understand the potential features of GPTs that are not present in quantum theory. Secondly, we wish to introduce a way in which intermediate determinism can be used as an axiom for quantum theory without ruling out two-level quantum systems and can be combined with the operational reasoning of GPTs.

In Sec.~\ref{sec:gisin} we summarise the relevant parts of the Ref.~\cite{gisin1991propensities} which inspired the present work. Sec.~\ref{sec:quantum} treats the case of quantum properties again but removes the dependency on the lattice structure. This allows us to generalise the principle of intermediate determinacy to quantum effects (which do not form a lattice) and recover the standard description of pure states. In Sec.~\ref{sec:gpts}, we introduce the relevant parts of the GPT framework. In Sec.~\ref{sec:gptprop} we define the natural analogue of intermediate determinism in GPTs and present our main result, Theorem~\ref{thm:main}, in which we identify the exact class of GPTs that obey this principle. We proceed by giving some examples of GPT systems that do and do not obey intermediate determinism in Sec.~\ref{sec:eg}. In this section we also give a corollary to our main result establishing exactly which GPTs with a Gleason-type theorem satisfy intermediate determinism. In Sec.~\ref{sec:gpmgpt} we find which GPTs obey the stronger requirement that the intermediate determinism follow from the effect space structure. Finally, Sec.~\ref{sec:prop} then briefly lays out how the concepts in this work allow intermediate determinism to be a possible axiom for deriving quantum theory in combination with the GPT framework without ruling out two-level quantum systems.


\section{Propensity}\label{sec:gisin}
In this section we will summarise the notion of intermediate determinism introduced by Gisin \cite{gisin1991propensities}. Gisin begins with the premise, established and motivated by Piron~\cite{piron1976foundations,piron1983new} and Aerts~\cite{aertsthesis,aerts1982description}, that the properties of a system must form a complete orthomodular lattice. A complete lattice is a partially ordered set $(L,\leq)$ such that any subset $K\subseteq L$ has a least upper bound 
denoted $\bigvee_{k\in K} k$, and a greatest lower bound 
denoted $\bigwedge_{k\in K} k$. We denote the greatest lower bound of $L$ by $0$, representing a property the system will never be found to possess. The partial order is interpreted as $l\leq k$ if property $l$ implies property $k$, i.e.~if property $l$ is actual so is property $k$. A lattice $L$ is orthomodular if it has a complement operation, denoted by $\cdot^c$, such that (i) $(l^c)^c=l$; (ii) $l\leq l^c$ only if $l=0$; (iii)  $l\leq k$ only if $l^c\geq k^c$; and, (iv) if $l\leq k$ then there exists $j\in L$ such that $j\leq l^c$ and $j\vee l=k$.  The complement $l^c$ is interpreted as the property of not having the property $l$. If $l\leq k^c$ we say that $l$ and $k$ are \emph{orthogonal} which is interpreted as the properties being disjoint but jointly measurable or testable and thus, never both being actual properties of a system in some state.

A state specifies the propensity of the system to actualise any given property from the lattice $L$. Thus, for each state there is a \emph{generalised probability measure} on the lattice of properties as defined below.

\begin{definition}\label{def:gpmp}
A generalised probability measure on a lattice of properties $L$, is a map $v:L\rightarrow[0,1]$ such that 
\begin{enumerate}[(i)]
\item for any sequence $(l_j)_j$ of pairwise orthogonal elements of $L$, 
\begin{equation}
v\left(\bigvee_jl_j\right)=\sum_j v(l_j)\,,
\end{equation} and,
\item $v(1)=1$, where $1\in L$ denotes the greatest element of $L$, $1=\bigvee_{l\in L} l$.
\end{enumerate}
\end{definition}

Gisin adds a third condition to this list to define a \emph{measure} on $L$:
\begin{itemize}
\item[($\star$)] for any subset $K\subset L$ such that $v(k)=1$ for all $k\in K$, $v(\bigwedge_{k\in K} k)=1$.
\end{itemize} 
Gisin notes that this condition has been critised and since it will not be relevant for the results of the present paper we exclude it from Def.~\ref{def:gpmp}.

For any measure $v$, let $l_v$ be the least actual property of $v$. Explicitly, $l_v=\bigwedge\{l\in L|v(l)=1\}$. Note that $l_v$ is guaranteed to also be an actual property by condition ($\star$). It follows that $l\in L$ is an actual property of $v$ if and only if $l_v \leq l$. With this framework in place we can impose the intermediate level of determinism by requiring that a state $v$ should be defined by its actual properties and hence its least actual property $l_v$. Hence any state can be represented by a \emph{propensity function}.

\begin{definition}\label{def:prop1}
A propensity function on a property lattice $L$ is a measure $v$ on $L$ with a unique least actual property, $l_v$, i.e.~if $w$ is a measure, we have $l_w=l_v$ if and only if $w=v$.
\end{definition}

Gisin \cite{gisin1991propensities} uses the definition of a state as a propensity function along with an axiom requiring that every (non-zero) property be actual for some state to narrow down the possible lattices of properties. The aim is to find that the property lattices must also be \emph{atomic} (and have at least four atoms) and satisfy the \emph{covering law}. At this point a result of Piron~\cite{piron1976foundations,gisin1991propensities,gisin84} would show that such property lattices always belong to classical or quantum systems. The axiom is shown to imply the lattice must be atomic and furthermore, all identified examples satisfy the covering law. However, it remains to be shown whether the covering law holds in general\footnote{In an earlier work~\cite{gisin84} the result of Piron was successfully used to single out classical and quantum theories by employing similar but stronger axioms.}. 

Although it is a possibility that only classical and quantum theories satisfy Gisin's axiom, we already know that this approach cannot \emph{rederive} quantum theory in its entirety since it rules out two-level or qubit systems. The property lattice of a system in classical physics is given by the power set of its phase space, $P(\Gamma)$. The partial order is given by inclusion and the complement is defined in the standard way, $X^c=\Gamma/X$. A propensity function on this lattice assigns probability one to a point, $x\in\Gamma$, of phase space and all the subsets containing that point. This assignment of probabilities coincides with the pure state described by the point $x$. It is clear that Gisin's axioms are satisfied.

On the other hand, quantum property lattices are given by the orthogonal projections $\PH$ onto closed subspaces of a given (complex) separable Hilbert space, $\cH$. When $\cH$ has dimension at least three Gleason's theorem ensures the existence of propensity functions. Namely, each propensity function is given by the Born rule for some pure state, i.e.~the functions $v(\Pi)=\langle\psi,\Pi\psi\rangle$ for some unit vector $\psi\in\cH$. However, Gleason's theorem does not hold in dimension two, and it follows that there exist no propensity functions on the lattice $\mathcal{P}(\C^2)$.


Explicitly, given any projection $\Pi\in\mathcal{P}(\C^2)$ then there are infinitely many generalised probability measures $v$ such that $l_v=\Pi$. Firstly, if $\Pi$ is rank-one, two possibilities are, $v_1(\Pi')=\Tr(\Pi'\Pi)$ and $v_2(\Pi')=2^{(\delta_{\Pi,\Pi'}-1)}-\delta_{\I-\Pi,\Pi'}/2$ for any rank-one projection $\Pi'$. Alternatively, if $\Pi$ is rank-two we have $\Pi=\I$, the identity operator on $\C^2$. Then $v(\Pi')=\Tr(\Pi'\rho)$ for any rank-two density operator $\rho$, gives $l_v=\I$. And finally, if $\Pi=0$ then for any $v$ such that $v(\Pi)=v(\Pi')=1$ for rank-one projections $\Pi\neq\Pi'$, we have in $l_v=\Pi$\footnote{This final set of measures is ruled out by Condition (2) on a measure in Ref.~ \cite{gisin1991propensities}. This difference leads to the same conclusion that there are no propensity functions on $\mathcal{P}(\C^2)$.}.

With no propensity functions, it is impossible for the lattice $\mathcal{P}(\C^2)$ to satisfy the requirement that every property is an actual property for some propensity function/state. Therefore, this axiom in its current formulation rules out two-level quantum systems. 

In contrast to Gleason's original theorem, a Gleason-type theorem does hold in dimension two \cite{Busch2003,Caves2004}. This Gleason-type theorem concerns \emph{generalised probability measures} on the set of quantum effects, $\EH$. Note that the Def.~\ref{def:gpmp} of a generalised probability measure does not rely upon every feature of a complete orthomodular lattice. We may therefore define a generalised probability measure on a more general structure such that Def.~\ref{def:gpmp} is recovered for complete orthomodular lattices, whilst the existing definition on quantum effects is also recovered. The generalised probability measure will also coincide with those considered in Gleason-type theorems for general probabilistic theories \cite{GPTGTT, farid2019}. The general structure and definition of a generalised probability measure is described in detail in Appendix~\ref{app:gpm}. For simplicity in the main text we will only state the resulting maps on the structures we study. 


\section{Quantum theory}\label{sec:quantum}
\subsection{Projections}\label{sec:proj}

In this section we reconsider the case of the property lattice of quantum theory, but without using the lattice structure of the projections $\PH$ so that we can generalise the definitions later to non-lattice structures.

The projections $\PH$ form a lattice with the ordering $\Pi\geq\Pi'$ if and only if $\Pi-\Pi'\in\PH$. The sets of maps given by Defs,~\ref{def:gpmp} and \ref{def:gpm} of a generalised probability measure\footnote{The definition of a measure from Ref.~\cite{gisin1991propensities} also gives the same set of maps as the extra condition ($\star$) is automatically satisfied on this lattice.} coincide on $\PH$ to give:
\begin{definition}\label{def:gpmproj}
A generalised probability measure $v$ on $\PH$ is a map $v:\PH\rightarrow[0,1]$ such that $v(\sum_j \Pi_j)=\sum_jv(\Pi_j)$ for all sequences of mutually orthogonal projections $(\Pi_j)_j\subset\PH$ and $v(\I_\cH)=1$.
\end{definition}

Propensity functions are intended to represent states that are uniquely identified by their actual properties. Since $\PH$ forms a complete lattice the set of actual properties of a state can be identified by its greatest lower bound. In the case of $\PH$ this greatest lower bound is also an actual property, referred to as the least actual property. However, since we are attempting to remove the dependence of our definitions on the lattice structure we will utilise the following less elegant but more direct concept of an \emph{actual set} in place of the least actual property. We define the actual set of a measure $v$ on $\PH$ is the subset $A_v=\left\{\Pi\in\PH|v(\Pi)=1\right\}$.

We may now give an alternative definition of a propensity function on $\PH$ which is equivalent to Def.~\ref{def:prop1}, and still captures the idea that these states are uniquely identified by their actual properties.

\begin{definition}\label{def:propproj}
A propensity function on $\PH$ is a generalised probability measure $v$ on $\PH$ such that for all measures $v'$ on $\PH$, $A_{v'}=A_v$ only if $v'=v$.
\end{definition}

Note that in a general lattice the greatest lower bound of the set of actual properties of a generalised probability measure is not necessarily an actual property itself. For this reason, the additional condition ($\star$) imposed on measures. However, since in Def.~\ref{def:propproj} the least actual property is replaced by the actual set, the condition ($\star$) is no longer necessary. For Hilbert spaces of dimension at least three the equivalence of the definitions of a propensity function can be seen by the following lemma.

\begin{lem}\label{lem:proj}
Let $\mathcal{H}$ be a separable Hilbert space with dimension at least three. Every propensity function $v:\PH\rightarrow[0,1]$ admits an expression 
\begin{equation}
v(\Pi)=\langle\psi,\Pi\psi\rangle\,,
\end{equation}
for all $\Pi\in\PH$ and some unit vector $\psi\in\mathcal{H}$.
\end{lem}
\begin{proof}
By Gleason's theorem, $v(\Pi)=\Tr(\Pi\rho)$ for some density operator $\rho$. In this case denote $A_v$ by $A_\rho$. Now we find $\Pi\in A_\rho$ if and only if $\Pi=\Pi_{\rm{supp}(\rho)}+\Pi'$, for the projection $\Pi'$ on to some subspace of $\rm{ker}(\rho)$, where $\Pi_{\rm{supp}}(\rho)$ is the orthogonal projection on to the support of $\rho$, as follows. 

Let $\rho=\sum_j\lambda_jP_j$ be a spectral decomposition of $\rho$ where $P_j$ are rank-one projections on to a subspace spanned by a unit vector $\psi_j$ and $\Pi\in A_\rho$. Since $0\leq\lambda_j\leq1$, $\sum_j\lambda_j=1$ and $0\leq\Tr(\Pi P_j)\leq1$, we find $\langle\psi_j,\Pi\psi_j\rangle=\Tr(\Pi P_j)=1$, for all $j$. The vectors $\psi_j$ span the support of $\rho$, therefore, we have $\Pi\chi=\chi$, for all $\chi\in\rm{supp}(\rho)$. The converse is clear.

Notice that $A_\rho$ is only a function of the support of $\rho$, hence any two density operators with the same support have the same actual set. Given a density operator $\rho$ with $\rm{ran}(\rho)\geq2$ there are infinitely many distinct density operators with the same support and therefore the same actual set. On the other hand, given a rank-one projection $P$ it is clear that $P\in A_\rho$ if and only if $\rho=P$, making $A_\rho$ unique to $\rho$. 

Thus $A_\rho=A_{\rho'}$ implies $\rho=\rho'$ if and only if $\rho=P$ for some rank-one projection $P$, in which case $v(\Pi)=\Tr(\Pi P)=\langle\psi,\Pi\psi\rangle$, for all $\Pi\in\PH$, where $\psi$ is a unit vector such that $P\psi=\psi$.  
\end{proof}

In dimension two the absence of propensity functions, and hence equivalence of Defs.~\ref{def:prop1} and~\ref{def:propproj}, follows from similar reasoning to that given in Sec.~\ref{sec:gisin}. Due to the fact that the pure states of a qubit are not propensity functions, we say that the property lattice of the system does not guarantee intermediate determinism. By this statement we mean that for each pure state of the system there exist other generalised probability measures with the same actual set as the pure state. If there were to be a state of the system with such a generalised probability measure (although this is not predicted by quantum theory) then the system would violate intermediate determinism. In higher dimensions, for pure states no such additional generalised probability measures exist and, thus, the intermediate determinism is guaranteed.

\subsection{Effects}\label{sec:eff}

We now generalise the notion of intermediate determinism from quantum properties to quantum effects. Consider a quantum system with Hilbert space $\cH$. Mathematically, the set of quantum effects $\EH$ comprises the self-adjoint operators $E$ on $\cH$ satisfying $0_\cH\leq E\leq I_\cH$, where $0_\cH$ and $I_\cH$ are the zero and identity operators on $\cH$, respectively, and $A\leq B$ means $\langle \psi,A \psi\rangle\leq \langle \psi,B \psi\rangle$ for all $\psi\in\cH$. Under this order quantum effects form a partially ordered set but not a lattice\footnote{For more detail on this point see, e.g.~Refs.~\cite{lahti1995partial,gudderlattice,Guddercounterexamples}}. Thus, effects are generally not interpreted as properties (see Sec.~\ref{sec:prop}).

Observables of a quantum system are most generally represented by positive-operator valued measures (POVMs). A POVM assigns a quantum effect to any subset of the set of values $\Omega$ of the observable. Precisely, it is a map ${\bf E}:\Sigma\to\EH$, where $\Sigma$ is a $\sigma$-algebra of subsets of $\Omega$, that is additive ${\bf E}(X\cup Y)={\bf E}(X)+{\bf E}(Y)$ for disjoint subsets $X$ and $Y$ and satisfies ${\bf E}(\Omega)=\I_\cH$.

In this way every possible outcome of measuring any observable of the system has an associated quantum effect. Assuming that measurement outcomes with the same effect occur with the same probability in any state motivates that states should map effects to probabilities via generalised probability measures, defined as follows \cite{Busch2003}.

\begin{definition}\label{def:qmeas}
A generalised probability measure $v$ on $\EH$ is a map $v:\EH\rightarrow[0,1]$ such that $v(\sum_j E_j)=\sum_jv(E_j)$ for every sequence of effects $(E_j)_j\subset\EH$ satisfying $\sum_jE_j\in\EH$ and $v(\I_\cH)=1$.
\end{definition}

We now say that an \emph{actual effect} of the system in a given state $v$ is an effect $E$ such that $v(E)=1$. In other words, any of the measurement outcomes, $X\in\Sigma$, of measuring an observable, ${\bf E}$, that are associated with the effect $E$ (i.e.~${\bf E}(X)=E$) occur with certainty when the system is in state $v$. We call the set of actual effects, $A_v=\left\{E\in\EH|v(E)=1\right\}$, of a generalised probability measure on $\EH$ its actual set.


Intermediate determinism for effects then means that the pure states of a system are uniquely determined by their actual set. Such a state is given by a propensity function on $\EH$.

\begin{definition}\label{def:propqeff}
A propensity function on $\EH$ is a generalised probability measure $v$ on $\EH$ such that for all measures $v'$ on $\EH$, $A_{v'}=A_v$ only if $v'=v$.
\end{definition}

Below we show that applying the property of intermediate determinism to the set of quantum effects $\EH$ for a given separable Hilbert space $\cH$ identifies the standard set of pure states in quantum theory, i.e.~the rays of $\cH$. Thus, intermediate determinism for effects recovers the result of Lemma~\ref{lem:proj} of intermediate determinism for properties but also applies in dimension two. This extension is due to the Gleason-type theorem of Busch~\cite{Busch2003} and Caves et al.~\cite{Caves2004} holding in dimension two.

\begin{lem}
Let $\mathcal{H}$ be a separable Hilbert space. Every propensity function $v:\EH\rightarrow[0,1]$ admits an expression 
\begin{equation}
v(E)=\langle\psi,E\psi\rangle\,,
\end{equation}
for all $E\in\EH$ and some unit vector $\psi\in\mathcal{H}$.
\end{lem}
\begin{proof}
By Busch's theorem, $v(E)=\Tr(E\rho)$ for some density operator $\rho$. Now we find $E\in A_\rho$ if and only if $E=\Pi_{\rm{supp}(\rho)}+E'$, for some effect $E'$ such that $\rm{supp}(E')\subseteq\rm{ker}(\rho)$, as follows. 

Let $E\in A_\rho$. Since $\rho$ is compact, we may write its spectral decomposition as $\rho=\sum_j\lambda_jP_j$ where $P_j$ are rank-one projections on to a subspace spanned by a unit vector $\psi_j$. Since $0\leq\lambda_j\leq1$, $\sum_j\lambda_j=1$ and $0\leq\Tr(EP_j)\leq1$, we find $\langle\psi_j,E\psi_j\rangle=\Tr(EP_j)=1$, for all $j$. As $\psi_j$ is a unit vector for all $j$ we find $E\psi_j=\psi_j$. The vectors $\psi_j$ span the support of $\rho$, therefore, we have $E\chi=\chi$, for all $\chi\in\rm{supp}(\rho)$. The converse is clear.

As with the projection case, notice that $A_\rho$ is only a function of the support of $\rho$, hence any two density operators with the same support have the same actual set. Given a density operator $\rho$ with $\rm{ran}(\rho)\geq2$ there are infinitely many distinct density operators with the same support and therefore the same actual set. On the other hand, given a rank-one projection $P$ it is clear that $P\in A_\rho$ if and only if $\rho=P$, making $A_\rho$ unique to $\rho$.

Thus $A_\rho=A_{\rho'}$ implies $\rho=\rho'$ if and only if $\rho=P$ for some rank-one projection $P$, in which case $v(\Pi)=\Tr(\Pi P)=\langle\psi,\Pi\psi\rangle$, for all $E\in\EH$, where $\psi$ is a unit vector such that $P\psi=\psi$. 
\end{proof}


\section{General probabilistic theories framework}\label{sec:gpts}
General probabilistic theories provide a family of operationally motivated physical theories with which to compare the quantum and classical theories that we believe describe nature. Such theories have been studied since the 1960s, with important early works including those of Mackey~\cite{mackeybook}, Ludwig~\cite{ludwig1967attempt,ludwighilbertspace} and Kraus~\cite{kraus1983states}. We will briefly summarise the GPT framework, in the formulation of Ref.~\cite{GPTGTT}, without the operational motivation. For more details on the modern formulation and motivation see Refs.~\cite{janotta2014generalized,BarrettGPT,Masanes2011,barnum2011information,Hardy2001a}.

A GPT describes a set of GPT systems (as quantum theory describes quantum systems). A GPT system has a \emph{state space} $\mathcal{S}$ is given by a convex, compact set of vectors of the form
\begin{equation}
\boldsymbol{\omega}=\begin{pmatrix}x_{1}\\
\vdots\\
x_{d}\\
1
\end{pmatrix}\in\R^{d+1}\,.\label{eq:probstates}
\end{equation}

As in quantum theory, each possible outcome of a measurement in a GPT is associated to an \emph{effect}. The set of all effects in a GPT system is known as its \emph{effect space} which will generally be denoted by $\mathcal{E}$.
The effect space $\mathcal{E}$ also corresponds to a convex subset of
$\R^{d+1}$. Given a state space $\cS$ every effect $\ef\in\E$ must satisfy $0\leq\ef\cdot\bo\leq1$ for all $\bo\in\cS$, since this number defines the probability of observing outcome $\ef$ after performing a suitable measurement on a system in state $\bo\in\cS$. The effect space necessarily contains the zero and unit vectors, 
\begin{equation}
\boldsymbol{0}=\begin{pmatrix}0\\
\vdots\\
0\\
0
\end{pmatrix}\qquad\text{ and }\qquad\boldsymbol{u}=\begin{pmatrix}0\\
\vdots\\
0\\
1
\end{pmatrix},\label{eq: zero and unit}
\end{equation}
as well as the vector $(\boldsymbol{u}-\boldsymbol{e})$ for every
$\boldsymbol{e}\in\mathcal{E}$ 
, which arises automatically as a valid effect. The effect
space also spans the full $(d+1)$ dimensions of the vector space.

\emph{Observables} (or \emph{meters}~\cite{Filippov2019}) are given by maps ${\bf E}:\Omega_{\bf E}\to\E$ from a countable outcome set $\Omega_{\bf E}=\{1,2,\ldots\}\subseteq\mathbb{N}$ such that $\sum_{x\in\Omega_{\bf E}}{\bf E}(x)=\bu$. An observable is often identified with its (ordered) image, a tuple $\left\llbracket \ef_{1},\ef_{2},\ldots\right\rrbracket $
of effects that sum to the unit effect $\boldsymbol{u}$,
such that each effect in the tuple represents a different possible
outcome when measuring the observable. In addition to the effect space, a GPT should specify the set of observables. Any valid set of observables contains all the couples $\left\llbracket \ef,\boldsymbol{u}-\ef\right\rrbracket $, along with the observables that are thereby \emph{simulable}~\cite{Heinosaarisimulable,Filippov2019} by means of taking classical mixtures of measurement procedures and post-processings of outcomes. The results in this paper are valid for all choices of sets of observables.

The \emph{no-restriction hypothesis} states that given a state space $\cS$ the effect space should comprise all possible effects, i.e.~$\E$ should be equal to the set
\begin{equation}\label{eq:ES}
E\left(\mathcal{S}\right) =\left\{ \boldsymbol{e}\in\R^{d+1}|0\leq\boldsymbol{e}\cdot\boldsymbol{\omega}\leq1,\text{ for all }\boldsymbol{\omega}\in\mathcal{S}\right\}\,.
\end{equation}
Analogously, given an effect space $\E$ the \emph{no-state-restriction hypothesis} says that the state space should contain all mathematically reasonable states, i.e.~should be given by the set
\begin{equation}\label{eq:we}
W(\E)=\left\{\bo\in\R^{d+1}\middle|\bo\cdot\ef\leq1\text{ for all }\ef\in\E\text{ and }\bo\cdot\bu=1\right\}\,.
\end{equation}
We will not assume either no-restriction hypothesis, however, we will find the maps in Eqs. \eqref{eq:ES} and \eqref{eq:we} very useful. For example, we will require the following result from Ref.~\cite{GPTGTT}.
\begin{lem}\label{lem:wes}
For any GPT with state space $\mathcal{S}$, we have $W\left(E\left(\mathcal{S}\right)\right)=\mathcal{S}$.
\end{lem}

We now give some simple examples of GPT systems: the classical-bit GPT, the NU bit and the aNU bit from Ref.~\cite{GPTGTT}, depicted in Fig.~\ref{fig:anu}. All three systems have state space, $\cS_B$ given by the line segment from $(-1,1)^T$ to $(1,1)^T$. The classical-bit has effect space, $\E_B=E(\cS_B)$, given by the square with vertices $\bO$, $\bu$ and $\ef_\pm=(\pm1/2,1/2)^T$. In the NU bit effect space the extremal points $\ef_\pm$ are replaced by scaled down versions $p\ef_\pm$ and their complements $\bu-p\ef_\pm$ for some $0<p<1$. The aNU bit has effect space, $\E_{aB}$, given by the intersection of two discs of radius $1/\sqrt{2}$ centred at $(\pm1/2,1/2)^T$. These examples demonstrate how the dual statement to Lemma~\ref{lem:wes} does not hold in general, since 
\begin{equation}
E(W(\E_{NB}))=E(W(\E_{aB}))=\E_B\neq\E_{NB}\neq\E_{aB}\,.
 \end{equation}

The NU bit is an example of a \emph{noisy unrestricted} (NU) GPT system, whereby the positive cone of the effect space is equal to the positive cone\footnote{For more details on the significance of cones in GPTs see Ref.~\cite{selby2021accessible}} of the unrestricted effect space, i.e.~$\E^+=E(\cS)^+$. Equivalently, for every $\ef\in\ES$ there exists $0<p\leq 1$ such that $p\ef\in\E$. The aNU bit does not satisfy this property but instead is an almost noisy unrestricted GPT system whereby the \emph{closure} of the positive cone of the effect space is equal to the positive cone of the unrestricted effect space, i.e.~$\overline{\E^+}=E(\cS)^+$. The class of aNU GPTs (including NU and unrestricted GPTs) is exactly the class of GPTs satisfying the no-state-restriction hypothesis~\cite{GPTGTT}.

\begin{figure}
\centering\includegraphics[width=0.75\textwidth]{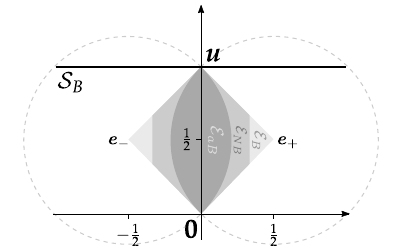}
\caption{\label{fig:anu} The bit state space $\cS_B$ and effect spaces $\E_B$, $\E_{NB}$ and $\E_{aB}$ of the classical bit, NU bit and aNU bit, respectively.}
\end{figure}

In this work will also require a few fundamental concepts from convex analysis. For a more detailed introduction see, e.g.~Refs.~\cite{Rockafellar1970,brondsted2012introduction}.
\begin{definition}
A \emph{closed half-space} of $\R^d$ is a set of vectors $\bx\in\R^d$ satisfying $\bh\cdot\bx\leq m$ for some $\bO\neq\bh\in\R^d$ and $m\in\R$.
\end{definition}
A \emph{supporting hyperplane} of a convex set, $S$, is the boundary of a closed half-space containing the set which also intersects the closure, $\overline{S}$, of the set.
\begin{definition}
A supporting hyperplane $H$ of a convex set $S$ in $\R^d$ is a set of points $\bx\in\R^d$ satisfying $\bh\cdot\bx=m$ for some $\bO\neq\bh\in\R^d$ and $m\in\R$ such that $\bh\cdot\boldsymbol{s}\leq m$ for all $\boldsymbol{s}\in S$ and $H\cap \overline{S}\neq\emptyset$.
\end{definition}
This intersection of the set with a supporting hyperplane constitutes an \emph{exposed face}.
An exposed face could be a point, a line segment or any other convex set. We define the dimension of a convex set via its \emph{affine hull}.
\begin{definition}
The affine hull of a subset $S\subseteq\R^d$ is the set of points $\bx=\lambda_1\boldsymbol{s}_1+\ldots+\lambda_n\boldsymbol{s}_n$ for real coefficients $\lambda_j$ such that $\lambda_1+\ldots+\lambda_n=1$ and $\boldsymbol{s}_1,\ldots,\boldsymbol{s}_n\in S$.
\end{definition} 
Now, we define the dimension of a convex set as the dimension of its affine hull. When an exposed face has dimension zero, i.e.~is a point, we call this point an \emph{exposed point}. Every exposed point $\bx$ is extremal (i.e.~$\bx$ cannot be written as a convex combination $\sum_jp_j\bx_j $ of other points from the set where $0<p_j<1$ for all $j$) but not every extremal point is exposed. The exposed points are, however, a dense subset of the extremal points.
 
Finally, the \emph{relative interior}, $\relint(S)$, of a convex set $S$ is the interior of the set $S$ when $S$ is viewed as a subset of its affine hull.
\begin{definition}
The relative interior, $\relint(S)$, of a convex set $S$ is the set of $\bx\in S$ such that for all $\by\in S$ there exists $\lambda>1$ such that $\lambda\bx+(1-\lambda)\by\in S$.
\end{definition}

\section{GPTs satisfying intermediate determinism}\label{sec:gptprop}

In Secs.~\ref{sec:proj} and \ref{sec:eff} we investigated how the intermediate determinism of quantum theory followed from the structure of properties (in dimensions greater than two) and effects. GPTs, on the other hand, may not even satisfy the principle of intermediate determinism, nevermind the stronger requirement that intermediate determinism must follow from the structure of the effect space. For example, in the aNU bit system in Sec.~\ref{sec:gpts}, the actual set of every state is simply the unit effect, thus none of the states have a unique actual set. In this section we show our main result, Theorem~\ref{thm:main} which states necessary and sufficient conditions for a GPT to satisfy intermediate determinism. In Sec.~\ref{sec:gpmgpt} we identify the subset of these GPTs in which intermediate determinism is a consequence of the effect space structure.

Thus, we will first identify exactly which GPTs satisfy intermediate determinism. A GPT satisfies intermediate determinism if each of its pure states (extremal points of its state spaces) can be uniquely identified by their sets of actual effects. 
Given a GPT state space $\cS$ with effect space $\E$, the actual set $A_\bo$ of a state $\bo\in\cS$ is given by
$A_\bo=\left\{\ef\in\E|\ef\cdot\bo=1\right\}$.

Explicitly, a GPT system satisfies intermediate determinism if for each pure state $\bo\in\cS$ we have that $A_{\bo'}=A_\bo$ implies $\bo'=\bo$ for all $\bo'\in\cS$. We will also sometimes need to consider the subset of the unrestricted effect space, $E(\cS)$, that gives probability one for a given state $\bo\in\cS$, in which case we will use the notation
\begin{equation}
A^{\ES}_\bo=\left\{\ef\in\ES|\ef\cdot\bo=1\right\}\,.
\end{equation}

To identify GPTs satisfying intermediate determinism, the exposed faces of GPT effect and state spaces will be important. In particular, we require the following two definitions.

\begin{definition}
An \emph{actual face} of a GPT effect space $\E$ is an exposed face $F\subset\E$ containing the unit effect $\bu$ that is maximal, in the sense that there does not exist an exposed face $F'$ of $\E$ such that $F\subsetneq F'$.
\end{definition}

\begin{definition}
A \emph{minimal exposed face} of a convex set is an exposed face $M$ such that $M\cap N$ either equals $M$ or the empty set, for all exposed faces $N$ of the convex set.
\end{definition}

We can now state the characterisation of GPTs satisfying intermediate determinism.
\begin{theorem}\label{thm:main}
A GPT system with state and effect spaces $\cS$ and $\E$ satisfies the principle of intermediate determinism if and only if 
\begin{enumerate}[(i)]
\item\label{itemF} for any pair of distinct actual faces $F$ and $F'$ of $E(\cS)$,  we have $F\cap\E\nsubseteq F'\cap\E$, and
\item\label{itemX} every extremal point of $\cS$ is exposed.
\end{enumerate}
\end{theorem}

The proof of this result will centre around the bijection between the minimal exposed faces of $\cS$ and the actual faces of $\ES$. To establish this relationship we give the following series of lemmata, with proofs in Appendix~\ref{app:proofs}. Firstly, we find that for every exposed face of a state space $\cS$ there is a point of $\ES$ which only appears in the actual sets of points in that face.

\begin{restatable}{lem}{sfaces}\label{lem:sfaces}
For every exposed face $G$ of $\cS$ there exists an element $\ff$ of $\ES$ such that $G=\{\bo\in\cS|\ff\cdot\bo=1\}$.
\end{restatable}

Conversely, we also find that every exposed face of an effect space $\E$ that contains the unit effect is the actual set of some point $\bo\in W(\E)$.

\begin{restatable}{lem}{ESfaces}\label{lem:ESfaces}
For every exposed face $F$ of $\E$ containing the unit effect there exists a vector $\bo\in W(\E)$ such that $F=\{\ef\in\E|\ef\cdot\bo=1\}$.
\end{restatable}

Note that setting $\E=\ES$ for any GPT state space $\cS$ in the lemma above gives that for every exposed face $F$ of $\ES$ containing the unit effect there exists a vector $\bo\in W(E(\cS))=\cS$ (by Lemma~\ref{lem:wes}) such that $F=\{\ef\in\E|\ef\cdot\bo=1\}$.

Now, we find that the actual set (in $\ES$) of every point in a given minimal exposed face of $\cS$ is an actual face of $\ES$.

\begin{restatable}{lem}{minmax}\label{lem:minmax}
Given a minimal exposed face $M$ of $\cS$ there exists an actual face $F$ of $\ES$ such that
\begin{enumerate}[(i)]
\item\label{item:F} $F=\{\ef\in\ES|\ef\cdot\bo=1\text{ for all }\bo\in M\}$, and
\item\label{item:A} $F=A^{\ES}_\bo$ for all $\bo\in M$.
\end{enumerate}
\end{restatable}

Next, we show that every actual face of $\ES$ is an actual set for some minimal exposed face of $\cS$. Thus, there is a bijection between the actual faces of $\ES$ and minimal exposed faces of $\cS$.

\begin{restatable}{lem}{maxmin}\label{lem:maxmin}
Given an actual face $F$ of $\ES$ there exists a minimal exposed face $M$ of $\cS$ such that $\ef\cdot\bo=1$ for all $\ef\in F$ and $\bo\in\cS$ if and only if $\bo\in M$.
\end{restatable}

Finally, we require one more lemma about state spaces in which every extremal point is exposed.

\begin{restatable}{lem}{exposed}\label{lem:exposed}
If all the extremal points of a convex set are exposed then every minimal exposed face of the set is a point.
\end{restatable}

We can now prove Theorem~\ref{thm:main}.

\begin{proof}
First, consider a GPT in which the state space $\cS$ and effect space $\E$ of each system satisfy Conditions \eqref{itemF} and \eqref{itemX}. Let $\bo_X$ be an extremal and hence exposed point of $\cS$. Then, by Lemma~\ref{lem:minmax} the actual set of $\bo_X$ in $\ES$ is an actual face $F$, i.e.~$A^{\ES}_{\bo_X}=F$. The actual set of $\bo_X$ in $\E$ is then $A^\E_{\bo_X}=F\cap\E$. Consider a point $\bo\in\cS$ such that 
\begin{equation}\label{eq:eqact}
A^\E_{\bo}=A^\E_{\bo_X}\,.
\end{equation}
 Let $\bo=\sum^N_jp_j\bo_j$ for some $N\in\mathbb{N}$ be a decomposition of $\bo$ as a convex combination of extremal points of $\cS$. We then find that 
\begin{equation}\label{eq:bigcap}
A^\E_\bo=\bigcap^N_jA^\E_{\bo_j}\,,
\end{equation}
since $\ef\cdot\bo=1$ if and only if $\ef\cdot\bo_j=1$ for all $j$. Since the extremal states $\bo_j$ are also exposed we find that their actual sets in $\ES$ are given by actual faces, $F_j$, of $\ES$, respectively. Then Eqs.~\eqref{eq:eqact} and \eqref{eq:bigcap} give 
\begin{equation}
\begin{aligned}
F\cap\E=A^\E_{\bo_X}=A^\E_\bo&=\bigcap^N_jA^\E_{\bo_j}\\
&=\bigcap^N_jF_j\cap\E\,.
\end{aligned}
\end{equation}
Thus, we have $F\cap\E\subseteq F_j\cap\E$ for all $1\leq j\leq N$ and by Condition~(\ref{itemF}) this implies $F_j=F$. It then follows from Lemmata~\ref{lem:exposed} and \ref{lem:maxmin} that $\bo_j=\bo_X$ for all $1\leq j \leq N$. Thus, we find $\bo=\bo_X$ and the GPT satisfies intermediate determinism.

Second, consider a system of a GPT satisfying intermediate determinism with state space $\cS$ and effect space $\E$. Let $\bo_X$ be an extremal point of $\cS$. To begin, we will show that $\bo_X$ must be a member of a minimal exposed face of $\cS$. Suppose, to the contrary, that $\bo_X$ is not contained in some minimal exposed face $M$ of $\cS$. Explicitly, let $M'$ be the intersection of all the exposed faces of $\cS$ containing $\bo_X$ and suppose $M'$ strictly contains some minimal exposed face $M\not\ni\bo_X$. Let $F$ be the subset of $\ES$ giving probability one for all states in $M'$, i.e.~$F'=\{\ef\in\ES|\ef\cdot\bo=1\text{ for all }\bo\in M'\}$. We will show $F'=A^{\ES}_{\bo_X}$. Assume there exists $\ff\notin F'$ such that $\ff\cdot\bo_X=1$. Then, $M''=\{\bo\in\cS|\ff\cdot\bo=1\}$ satisfies $M''\cap M'\neq M'$ contradicting the definition of $M'$. 

Now, by Lemma~\ref{lem:minmax} the states of the minimal exposed face $M\subsetneq M'$ all share an actual set consisting of an actual face $F$ of $\ES$. Since $M\subsetneq M'$ we have $F'\subsetneq F$. Now consider an equal mixture $\bo'$ of $\bo_X$ with a vector $\bo_M$ of $M$. The actual set of this mixture is also $F'$ and hence $\bo_X$ would be an extremal state without a unique actual set.

We have shown that each extremal point $\bo_X$ of $\cS$ must be a member of a minimal exposed face, $M$. We will now show that $\bo_X$ must itself be a minimal exposed face, i.e.~we will show Condition~(\ref{itemX}) to hold. By Lemma~\ref{lem:minmax}, the actual set $A^{\ES}_{\bo_X}$ in $\ES$ is an actual face and we have $A^\E_{\bo_X}=A^{\ES}_{\bo_X}\cap\E$. Let $\bo$ be a generic point in the minimal exposed face $M$ containing $\bo_X$. Lemma~\ref{lem:minmax} tells us that $A^{\ES}_\bo=A^{\ES}_{\bo_X}$. Therefore, we find 
\begin{equation}
A^{\E}_\bo=A^{\ES}_\bo\cap\E=A^{\ES}_{\bo_X}\cap\E=A^{\E}_{\bo_X}\,,
\end{equation}
and thus, $\bo=\bo_X$ by the intermediate determinism of the GPT. Thus, Condition~(\ref{itemX}) must hold.

Finally, combining Lemma~\ref{lem:maxmin} with Condition~(\ref{itemX}) we find that given a pair of actual faces $F$ and $F'$ of $\ES$ there exists a pair of extremal points $\bo$ and $\bo'$ of $\cS$ such that $F$ and $F'$ are the actual sets of $\bo$ and $\bo'$ in $\ES$, respectively. Suppose $F\cap\E\subseteq F'\cap\E$ and let $\bo_m$ be an equal mixture of $\bo$ and $\bo'$. Then, $A^\E_\bo=F\cap\E=A^\E_{\bo_m}$ and the intermediate determinism of the GPT gives $\bo_m=\bo=\bo'$. Thus, we find $F=F'$ and Condition~(\ref{itemF}) holds.   

\end{proof}

\section{Examples of GPTs with and without intermediate determinism}\label{sec:eg}

The two conditions (\ref{itemF}) and (\ref{itemX}) of Theorem~\ref{thm:main} together are neither necessary nor sufficient for a GPT to satisfy the no-restriction hypothesis or admit a \emph{Gleason-type theorem}. We say a GPT admits a Gleason-type theorem when for each system every generalised probability measure on the effect space is given by the inner product with a point in the state space. This condition is equivalent to $\cS=W(\E)$~\cite{GPTGTT,farid2019}, where $\cS$ and $\E$ are the state and effect spaces of the system, respectively. In terms of restrictions on the unrestricted effect space $\ES$, a GPT system has a Gleason-type theorem if and only if the restriction is \emph{almost noisy}, explicitly when the effect space satisfies $\overline{\E^+}=E(\cS)^+$. Satisfying the no-restriction hypothesis is sufficient but not necessary for a GPT to admit a Gleason-type theorem.

Obeying the no-restriction hypothesis ($\E=\ES$) is also not sufficient for a GPT system to satisfy intermediate determinism since the state space $\cS$ could have extremal points that are not exposed. A classic example of such a convex set is a ``pill'' shape given by a square with two semicircles attached to two opposite sides, as in Fig.~\ref{fig:freak}. Explicitly, this set is the convex hull of two arcs given by 
\begin{equation}
\begin{pmatrix}
\cos(\theta)+1\\
\sin(\theta)\\
1
\end{pmatrix}\text{ for }-\frac{\pi}{2}\leq\theta\leq\frac{\pi}{2}\text{, and }
\begin{pmatrix}
\cos(\theta)-1\\
\sin(\theta)\\
1
\end{pmatrix}\text{ for }\frac{\pi}{2}\leq\theta\leq\frac{3\pi}{2}\,.
\end{equation}

Taking this convex subset of $\R^3$ to be the state space $\cS_p$ and letting $\E_p=E(\cS_p)\subset\R^3$ be the effect space results in a GPT system with a Gleason-type theorem, satisfying the no-restriction hypothesis that does not satisfy intermediate determinism.

\begin{figure}\centering
\includegraphics[scale=0.75]{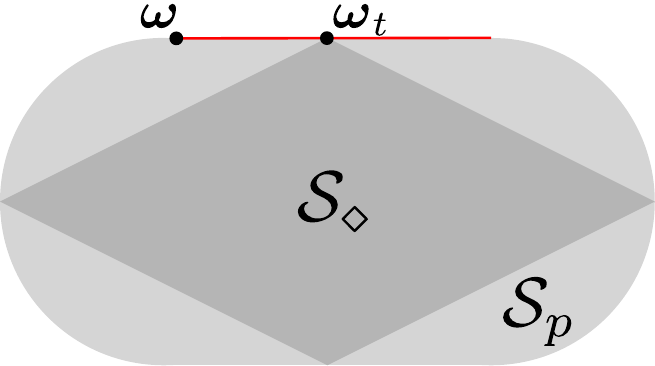}
\caption{\label{fig:freak} The state space $\cS_p$ in which not all extremal points are exposed, for example, the point marked $\omega$. A minimal exposed face is indicated in red. The darker diamond depicts the restricted state space, $\cS_\diamond$, with one of its four extremal points $\bo_t$ highlighted.}
\end{figure}

\begin{figure}
\includegraphics[width=\textwidth]{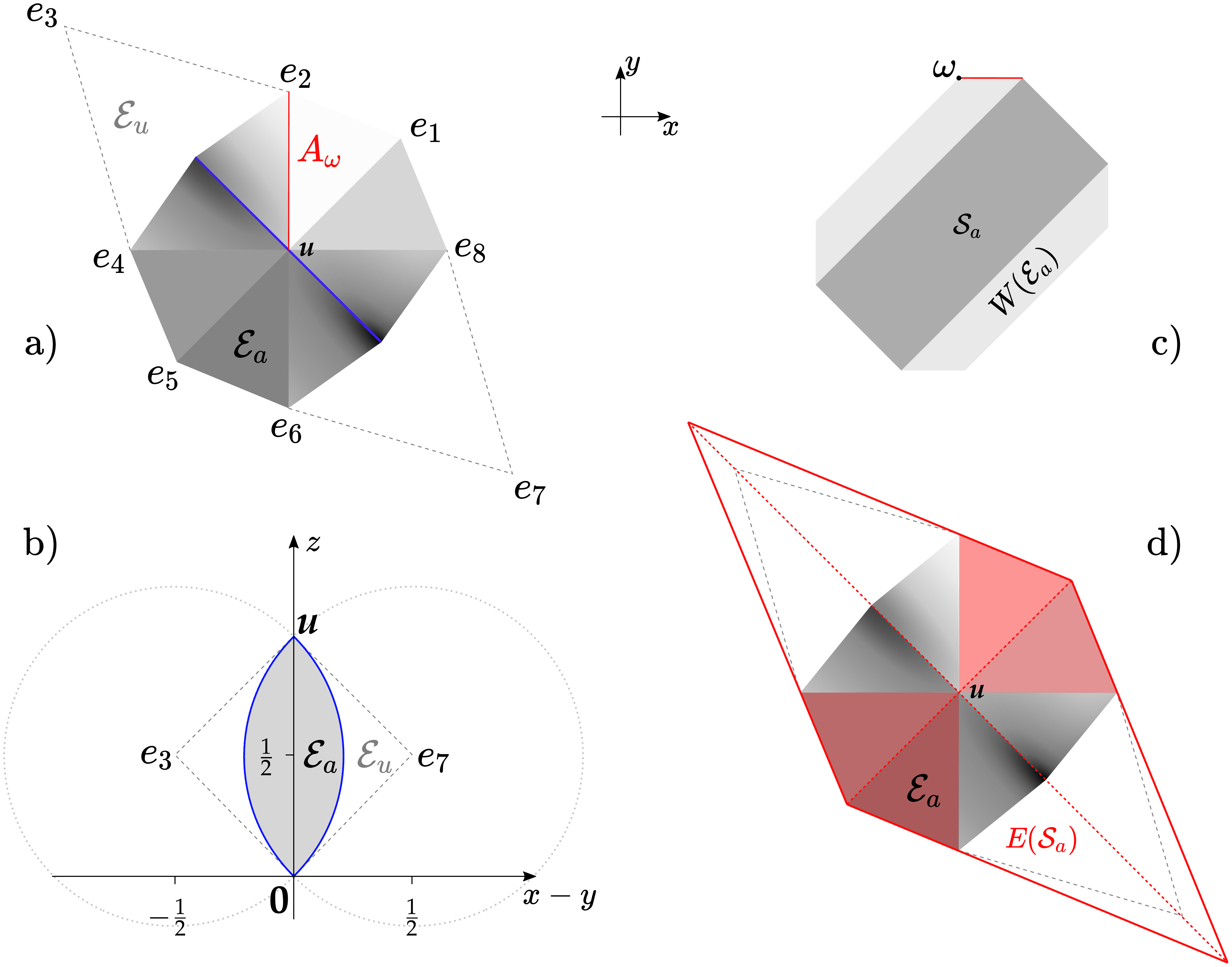}
\caption{\label{fig:oct}Depiction of the state space $\cS_a$, the effect space $\E_u$ and an almost noisy restriction $\E_a$ that satisfies intermediate determinism. Diagrams a), c) and d) show a ``bird's-eye'' view where the $z$-axis is normal to the page. Diagram a) shows the effect space $\E_a$ as a restriction of polyhedral effect space $\E_u$ given by the convex hull of the eight effects $e_1$ to $e_8$ and the zero and unit effects. In the restricted effect space $\E$ the extremal points $e_3$ and $e_7$ are discarded and replaced by continuous arcs of extremal points shown in blue. A cross-section of the effect space along the blue line in a) is shown in b). Diagram c) shows the state space $\cS_a$ compared to the unrestricted state space $W(\E_a)$. Finally, diagram  d) shows the effect space $\E_a$ compared to the unrestricted effect space $E(\cS_a)$, where the four shaded red triangles indicate the intersections of each of the four actual faces of $E(\cS_a)$ with $\E_a$.}
\end{figure}

We can also find examples of GPT systems that violate the no-restriction hypothesis and intermediate determinism but still have a Gleason-type theorem. The aNU bit system described in Sec.~\ref{sec:gpts} and Fig.~\ref{fig:anu} is such a system. 

In an almost NU GPT that is not a NU GPT, such as the aNU bit, the positive cone of the effect space is not closed, i.e.~$\overline{\E^+}\neq \E^+$. All such effect spaces have extremal effects arbitrarily close to the zero and unit effects. Unlike the aNU bit, effect spaces of this type can be part of GPTs satisfying intermediate determinism if the state spaces are restricted in the correct way. For example, consider an effect space $\E_u$ in $\R^3$ given by the convex hull of the zero and unit effects, $\bO$ and $\bu=(0,0,1)^T$ and eight extremal points arranged in a octagon\footnote{This system is a modified version of the octogon system from Ref.~\cite{janotta2011limits}.} in the plane $(x,y,1/2)^T$ for $x,y\in\R$, given by
\begin{equation}
e_j=\frac12\begin{pmatrix}
\cos \frac{\pi j}{4}\\\sin \frac{\pi j}{4}\\1
\end{pmatrix}\text{ for }j=3,7\text{ and }
e_j=\frac14\begin{pmatrix}
\cos \frac{\pi j}{4}\\\sin \frac{\pi j}{4}\\2
\end{pmatrix}\text{ otherwise.}
\end{equation}
Then take an almost-noisy restriction, $\E_a$, of this effect space by replacing $e_3$ and $e_7$ by continuous arcs of extremal effects depicted in blue on Fig.~\ref{fig:oct}~a) and b). These arc of extremal effects are chosen such that $\E_a^+\neq\E_u^+$ but $\overline{\E_a^+}=\E_u^+$. 

Now $\E_a$ together with the unrestricted state space $W(\E_a)=W(\E_u)$ shown in Fig.~\ref{fig:oct}~c) does not satisfy intermediate determinism. For example, the extremal point $\omega$ of $W(\E_a)$ has an actual set $A_\omega$ given by the convex hull of $e_2$ and $\bu$ depicted by the red line in Fig.~\ref{fig:oct}~a). However, all the states in the face of $W(\E_a)$ highlighted by the red line in Fig.~\ref{fig:oct}~c) (minus the other end point) also have the same actual set $A_\omega$. On the other hand, the state $\omega$ would have actual set given by the convex hull of $e_2$, $e_3$ and $\bu$ in the unrestricted effect space $\E_u$.

Out of the eight extremal points of $W(\E_a)$, four have the problem of no longer begin propensity functions. However, we may remove them without introducing any further extremal points (or reducing the dimension of the state space) by taking the convex hull of the four propensity functions, resulting in the restricted state space $\cS_a$ shown in Fig.~\ref{fig:oct}~c). 

The pair $\E_a$ and $\cS_a$ exhibit intermediate determinism since they satisfy requirements (\ref{itemF}) and (\ref{itemX}) of Theorem~\ref{thm:main}. In particular, in Fig.~\ref{fig:oct}~d) we see how the intersection of each of the four actual faces of $E(\cS_a)$ with $\E_a$ is a distinct actual face of $\E_a$.

We conclude this section by noting that if all the extremal points of a GPT's state spaces are exposed then the no-restriction hypothesis is sufficient to ensure that intermediate determinism holds, as is being a NU GPT. This statement is a corollary of Theorem~\ref{thm:main}.

\begin{corollary}
A noisy unrestricted GPT in which all the extremal points of the state spaces are exposed obeys intermediate determinism.
\end{corollary}

\begin{proof}
We need to show that in a NU GPT system with state space $\cS$ and effect space $\E$ satisfy Condition~(\ref{itemF}) of Theorem~\ref{thm:main}. Let $F$ and $F'$ be two distinct actual faces of $E(\cS)$. Then, given $\ef\in F/F'$ we have that $p(\bu-\ef)\in\E$ for some $0<p\leq 1$ and thus, $\ef_p=\bu-p(\bu-\ef)=p\ef-(1-p)\bu\in\E$. Let $\bo$ be a state for which $F'$ is the actual set, i.e., $\ff\cdot\bo=1$ for all $\ff\in F'$. The effect $\ef_p$ is clearly in $F$ however, cannot be in $F'$ since this would require $\ef\cdot\bo=1$ and therefore, $\ef\in F'$.
\end{proof}

It follows that most of the GPTs in the literature satisfy intermediate determinism, such as all the polytope systems from Ref.~\cite{janotta2011limits}. However, being a NU GPT is not necessary, for example the convexification~\cite{janottanorestriction} of the Spekkens' toy theory~\cite{spekkens2007evidence} is not a NU GPT but does obey intermediate determinism.

\section{Intermediate determinism from effect space structure}\label{sec:gpmgpt}

Now we have established exactly which GPTs satisfy the principle of intermediate determinism, we can also identify in which GPTs intermediate determinism is guaranteed by the structure of the effect space, as is the case with quantum theory. This property is necessary if one wishes to have an axiomatisation of a theory in which pure states are propensity functions. Recall the example of the qubit lattice of \emph{properties}. Here we said that although the system satisfies intermediate determinism, in the sense that each pure state has a unique actual set, this fact is not guaranteed by the property lattice. Explicitly, for every pure state there exist other generalised probability measures (not given by density operators) with the same actual set. Thus, there exist no propensity functions on this lattice. In particular, the pure states of a qubit are not propensity functions on its lattice of properties. 

On the other hand, we found that the pure states of a qubit were propensity functions on the qubit \emph{effect space}. Therefore, assuming the effect space structure there is no other theoretically possible state with the same set of actual properties as any given pure state. Thus, we say the effect space structure guarantees the intermediate determinism of the qubit. We now define generalised probability measures and propensity functions on a GPT effect space in order to identify in which GPTs the pure states are propensity functions on the effect spaces.

The following definition of a generalised probability measure on a GPT effect space coincides with Def.~\ref{def:qmeas} for quantum effects when the set of quantum effects is viewed as a GPT effect space. The definition is also a special case of Def.~\ref{def:gpm}. 
\begin{definition}\label{def:gpmgpt}
A generalised probability measure $v$ on a GPT effect space $\E$ is a map $v:\E\rightarrow[0,1]$ satisfying 
\begin{equation}
v(\ef+\ff+\ldots)=1\,
\end{equation}
for all sequences of effects $(\ef,\ff,\ldots)\subset\E$ such that $\left\llbracket\ef,\ff,\ldots\right\rrbracket$ is an observable.
\end{definition}
It was shown in Ref.~\cite{GPTGTT} that for any GPT effect space and any set of observables every generalised probability measure $v$ can be expressed as $v(\ef)=\ef\cdot\bo$ for some $\bo\in W(\E)$.

 We can similarly generalise the notion of an actual set to a generalised probability measure, $v$, on a GPT effect space, as the set of effects, $A_v=\left\{\ef\in\E|v(\ef)=1\right\}$, which occur with certainty when the system is in a state with generalised probability measure $v$.  


Now, we define a propensity function as a generalised probability measure that is uniquely identified by its actual set. 

\begin{definition}
A propensity function on a GPT effect space $\E$ is a generalised probability measure $v$ on $\E$ such that $A_v=A_{v'}$ implies $v=v'$ for all generalised probability measures $v'$ on $\E$.
\end{definition}

Thus, if the extremal states in a GPT system are propensity functions on the effect space there cannot exist other generalised probability measures with the same actual set. In general, the set of generalised probability measures can strictly contain the state space. Therefore, all pure states being propensity functions is a stronger statement than the GPT satisfying intermediate determinism, whereby all pure states must have a unique actual set amongst all the states in the state space but not necessarily amongst the larger set of generalised probability measures. 

For example, consider the state space $\cS_\diamond$ depicted in Fig.~\ref{fig:freak} which is a restriction of the pill shaped state space, $\cS_p$, along with the unrestricted effect space $E(\cS_p)$ of $\cS_p$. The pair $\cS_\diamond$ and $E(\cS_p)$ satisfy intermediate determinism but the top extremal point, $\bo_t$, is not a propensity function. The section of the boundary marked in red, of which $\bo_t$ is the midpoint, is a minimal exposed face of $\cS_p$ and thus, by Lemma~\ref{lem:minmax}, all the points in this face have the same actual set in $\E(\cS_p)$.

The propensity functions on an effect space are given by the points of $W(\E)$ whose actual set is an actual face of $\E$. This fact can be seen from Lemma~\ref{lem:ESfaces} and the fact that the actual set of any state in $W(\E)$ is an exposed face of $\E$ containing $\bu$. It follows from Lemmata~\ref{lem:minmax} and \ref{lem:maxmin} that all the propensity functions will be exposed points of $W(\E)$ but not all exposed points are propensity functions. Therefore, in order for intermediate determinism to be guaranteed by the effect space structure we need that the extremal points of the state space $\cS$ are a subset exposed points of $W(\E)$ with actual sets given by actual faces of $\E$. More precisely:

\begin{obs}\label{obs:props}
All the extremal states of a GPT system are propensity functions if and only if they are a subset of the (exposed) points of $W(\E)$ whose actual sets are actual faces of the effect space of the system.
\end{obs}

\section{Properties in GPTs}\label{sec:prop}

Propensity functions were originally introduced to represent states that were completely characterised by their actual \emph{properties}. In Secs.~\ref{sec:eff} and~\ref{sec:gpmgpt} we generalised propensity functions to act on quantum/GPT effect spaces. In this setting propensity functions represent states that are completely characterised by their actual \emph{effects}. Effects are generally not interpreted as representing properties of a quantum system. In order to do so much of the understanding of what is meant by a property would be lost since effects do not form a lattice.

The generalised notion of propensity functions for effects captures the meaning of intermediate determinacy but applies it to the likelihood of measurement events as opposed to the propensity of a system to take certain properties. The idea is made more operational and, thus, becomes compatible with the GPT framework.

The advantage of generalising to effects is that propensity functions now exist for two-level quantum systems. The disadvantage is that since we are no longer considering properties the possibility for deriving quantum and classical theories using the result of Piron (as summarised in Sec.~\ref{sec:gisin}) is lost. However, if we reintroduce the notion of a property into the effect paradigm this argument could become viable once more\footnote{Properties and probabilistic theories have been treated simultaneously before, for example in the test-space formalism \cite{testspaceprobs}.}. 

One candidate for the properties of a GPT system is its set of extremal effects. In this case of quantum theory this definition recovers the property lattice $\PH$~\cite{davies1976quantum}. The extremal effects are endowed with a natural partial order given by $\ef\leq \ff$ if $\ff-\ef\in\E^+$. This partial order is compatible with the interpretation that if $\ef\leq\ff$ and $\ef$ is actual then so is $\ff$. 

We may then use the ideas of the present manuscript to continue following the path laid out by Gisin to try single out quantum theory, but without ruling out two-level quantum systems. Namely, we can combine three requirements (i) properties form a complete orthomodular lattice, (ii) states are generalised probability measures on effect spaces and pure states are also propensity functions on the properties and (iii) every (non-zero) property is an actual property of some state. Note that requirements (i) and (iii) are follow those of Gisin and only (ii) has been modified.

We can then discover whether there exist complete orthomodular lattices of GPT properties satisfying intermediate determinism besides the quantum and classical cases, or subtheories thereof. We can also investigate how the characteristics of a GPT property lattice translate into features of the effect space, and vice versa. For example, in which effect spaces would the extremal effects form a lattice satisfying the covering law? And what does an effect space being an almost noisy restriction (satisfying $\E^+\neq\overline{\E^+}$) mean for the structure of extremal effects?

\section{Discussion}
In this work we have generalised the idea of intermediate determinism from properties to measurement events. Under this generalisation we have found that the propensity functions on quantum effect spaces are given by the Born rule given some pure quantum state. Furthermore, we have identified exactly which GPTs satisfy intermediate determinism in our main result, Theorem~\ref{thm:main} and also when the pure states of a GPT system are given by propensity functions in Observation~\ref{obs:props}. 

In future work the author wishes to explore introducing the idea of properties into the GPT framework, as proposed in Sec.~\ref{sec:prop}. Then the results of this manuscript could be used to attempt to single out quantum theory using the method proposed by Gisin~\cite{gisin1991propensities} without ruling out quantum systems in dimension two. As suggested in Sec.~\ref{sec:prop}, one candidate for properties that recovers the property lattice $\PH$ from the quantum effect space is the set of extremal effects. Alternatively, in this work the exposed faces of the effect spaces were shown to play an important role in intermediate determinism. These faces are an appealing candidate since the exposed faces of any convex set naturally form a complete lattice.

Introducing properties to GPTs will also open up questions such as, whether there exist complete orthomodular lattices of GPT properties satisfying intermediate determinism or the covering law other than those of quantum and classical theory. These structures may also shed light on the open question of Gisin, as to whether there exist atomic, complete, orthomodular lattices with at least four orthogonal atoms that do not satisfy the covering law.

\section*{Acknowledgements}
The author would like to thank Nicolas Gisin for pointing out the role of Gleason's theorem in Ref.~\cite{gisin1991propensities} and for helpful comments. This project has received funding from the European Union's Horizon 2020 research and innovation programme under the Marie Sk\l odowska-Curie grant agreement No.~754510. The author acknowledges support from the Government of Spain (FIS2020-TRANQI and Severo Ochoa CEX2019-000910-S), Fundaci\'o Cellex, Fundaci\'o Mir-Puig, Generalitat de Catalunya (CERCA, AGAUR SGR 1381).

\bibliographystyle{unsrturl}
\bibliography{propensitybib2}
\appendix
\section{Generalised probability measure structure}\label{app:gpm}
In this appendix we give a definition of a generalised probability measure on a very general structure. Our definition coincides with the definitions of a generalised probability measures on a orthomodular lattices, quantum effect spaces and GPT effect spaces given in Defs.~\ref{def:gpmp}, \ref{def:gpmproj}, \ref{def:qmeas} and \ref{def:gpmgpt}. 

Let $(M,\oplus)$ be partial commutative monoid, with an element $u$ such that $u\oplus m$ is defined if and only if $m=0$ where zero is the identity element of the operation $\oplus$. Explicitly, $M$ is a set with a partial binary operation $\oplus$ such that:
\begin{enumerate}[(i)]
\item if $m\oplus n$ is defined then $n\oplus m$ is defined and $m\oplus n=n\oplus m$;
\item if $(m\oplus n)\oplus o$ is defined then $m\oplus (n\oplus o)$ is defined and we denote $m\oplus n \oplus o\vcentcolon=m\oplus (n\oplus o)=(m\oplus n)\oplus o$;
\item there exists an identity element $0\in M$ such that $0+m=m$ for all $m\in M$;
\item there exists an element $u\in M$ such that $u\oplus m$ is defined if and only if $m=0$.
\end{enumerate} 
Additionally, we say that $\bigoplus_{j=1}^\infty m_j$ is defined if $\bigoplus_{j=1}^N m_j$ for all $N\in\mathbb{N}$. The existence of $m\oplus n$ generalises the idea $m$ and $n$ being jointly measurable as effects or orthogonal as properties.

\begin{definition}\label{def:gpm}
A generalised probability measure $v$ on $(M,\oplus)$ is a map $v:M\rightarrow[0,1]$ such that (i) for any sequence $(m_j)$ such that $\bigoplus_{j=1}^\infty m_j$ is defined, we have $v(\bigoplus_{j=1}^\infty m_j)=\sum_{j=1}^\infty v(m_j)$, and (ii) $v(u)=1$.
\end{definition}

If $M$ forms a complete orthomodular lattice, we take the partial operation $\oplus$ to be the least upper bound operation restricted to elements $a,b$ such that $a<b^c$. For quantum and GPT effect spaces we take $\oplus$ to be the standard addition operation, restricted to pairs of effects that sum to an effect. Under these mappings the definitions of a generalised probability measure coincide with Defs.~\ref{def:gpmp}, \ref{def:gpmproj}, \ref{def:qmeas} and \ref{def:gpmgpt}. The only non-immediate aspect to verifying this fact is showing that the infinite sum $\bigoplus_{j=1}^\infty m_j$ can be defined as we would wish in the lattices/effect spaces. For example, in a quantum effect space we will demonstrate that the infinite sum $\bigoplus_{j=1}^\infty E_j$ is defined exactly when $\left(\sum_{j=1}^N E_j\right)$ weakly converges to an effect, and therefore taking $\bigoplus_{j=1}^\infty E_j=\text{w-}\lim_{N\to\infty}\sum_{j=1}^N E_j$ is well-defined, where $\text{w-}\lim$ is the weak limit.

Let $(E_j)$ be a sequence of effects in $\EH$ such that $S_N=\sum_{j=1}^NE_j\leq E\in\EH$ for all $N\in\mathbb{N}$. Then we have $(S_N)$ is increasing and bounded above therefore (for example, by Theorem 2.7 in Ref.~\cite{busch2016quantum}), the set $\{S_N|N\in\mathbb{N}\}$ has a least upper bound $S$ to which $(S_N)$ weakly converges. Since $E$ is an upper bound of $\{S_N|N\in\mathbb{N}\}$, we find $\sum_{j=1}^\infty E_j=S\leq E$ is an effect. For the converse, clearly if  $(E_j)$ is a sequence of effects in $\EH$ such that $\sum_{j=1}^N E_j$ weakly converges to an effect then $\sum_{j=1}^N E_j\in\EH$ for all $N\in\mathbb{N}$.

\section{Proofs of Lemmata~\ref{lem:sfaces}, \ref{lem:ESfaces}, \ref{lem:minmax}, \ref{lem:maxmin} and \ref{lem:exposed}}\label{app:proofs}
\sfaces*
\begin{proof}
Firstly, if $G=\cS$ then we may take $\ff=\bu\in\ES$. 

Otherwise, given a state space $\cS$ embedded in $\R^{d+1}$, let $\cS_d$ denote the $d$-dimensional state space before the embedding. Explicitly, 
\begin{equation}\cS=\{(\bx_1,\ldots,\bx_d,1)^T|(\bx_1,\ldots,\bx_d)^T\in\cS_d\}.
\end{equation}
Let $H_d$ be the supporting hyperplane of $\cS_d\subset\R^d$ intersecting $\cS_d$ at the exposed face $G_d$. Then we may describe $H_d$ as the set of vectors $\by_d\in\R^d$ such that
\begin{equation}\label{eq:yd}
\bh_d\cdot\by_d=m,
\end{equation}
for some vector $\bh_d=(h_1,h_2,\dots,h_d)^T\in\R^d$ and real number $m\in\R$, chosen such that $\bh_d\cdot\bo_d\leq m$ for all $\bo_d\in \cS_d$.
Now, letting $\bh=\bh_d\oplus(1-m)\in\R^{d+1}$ and $\by=\by_d\oplus1\in\R^{d+1}$, Eq. \eqref{eq:yd} may be equivalently written as 
\begin{equation}\label{eq:y}
\bh\cdot\by=1\,,
\end{equation}
and $\bh\cdot\bo\leq1$ for all $\bo\in \cS$.

Finally, let $l=\min_{\bo\in\cS}\bh\cdot\bo$. If $l\geq0$, setting $\ff=\bh$ gives $\ff\in\ES$ and $\{\bo\in\cS|\ff\cdot\bo=1\}=G$. If $l<0$ we may mix $\bh$ with the unit effect to bring its minimum value on $\cS$ to zero whilst preserving the Eq. \eqref{eq:y}. Explicitly, set 
\begin{equation}
\ff=\frac{l}{l-1}\bu+\frac{1}{1-l}\bh\,.
\end{equation}
Now, similarly the vector $\ff\in\R^{d+1}$ satisfies $\ff\in\ES$ and $\{\bo\in\cS|\ff\cdot\bo=1\}=G$.
\end{proof}

\ESfaces*
\begin{proof}
Let $H=\{\ef\in\E|\ef\cdot\bh=x\}$ be a supporting hyperplane of $\E$ intersecting $\E$ at $F$, where $\bh\in\R^{d+1}$ and $x\in\R$.

Firstly, we will show that $x\neq0$ since $\E$ must span $\R^{d+1}$ and contain $\bu-\ef$ for each $\ef$ it contains. Assume $x=0$ and, w.l.o.g., $\bh\cdot\ef\geq0$ for all $\ef\in\E$. Then, since for each $\ef\in\E$ we have $\bu-\ef\in\E$ we find $\bh\cdot(\bu-\ef)\geq0$ and therefore, $\bh\cdot\ef\leq0$. Thus, we have that $\bh\cdot\ef=0$ for all $\ef\in\E$ meaning $\E$ could not span $\R^{d+1}$.

Since $x\neq0$ we may define $\bo=\bh/x$. We find if $\ef\in H$ then $\ef\cdot\bo=1$. It follows that $\bu\cdot\bo=1$ since $\bu\in F\in H$. Finally, we also have $\ef\cdot\bo\leq1$ for all $\ef\in\E$. In turn this gives $(\bu-\ef)\cdot\bo\leq1$ and thus, $\ef\cdot\bo\geq0$ for all $\ef\in\ES$. Therefore, we have $\bo\in W(\E)$. 
\end{proof}

\minmax*
\begin{proof}
For Statement~(\ref{item:F}), let $M$ be a minimal exposed face of $\cS$. Then there exists $\ff\in\ES$ such that $M=\{\bo\in\cS|\ff\cdot\bo=1\}$ by Lemma~\ref{lem:sfaces}. Let $\bm\in\relint(M)$ and $F=\{\ef\in\ES|\ef\cdot\bm=1\}$. The set $F$ is an exposed face of $\ES$ contained $\bu$, now we must show that it is maximal.

Let $F'$ be the actual face containing $F$ and $\ff'\in\relint(F')$. Then we define $M'=\{\bo\in\cS|\ff'\cdot\bo=1\}$. We will show $M'=M$. Let $\bo'\in M'$ then, by definition, we have $\ff'\cdot\bo'=1$. Now, since $\ff'$ is in the relative interior of $F'$ we have that $\ef\cdot\bo'=1$ for all $\ef\in F'$ including $\ff\cdot\bo=1$. By the definition of $M$ it follows that $\bo'\in M$ and we can conclude that $M'\subseteq M$. However, $M'$ is an exposed face of $\cS$ and $M$ is minimal therefore we have $M'=M$.

Finally, given $\ef\in F'$ we have $\ef\cdot\bm=1$ since $\bm\in\relint(M)=\relint(M')$. Thus, by the definition of $F$ we also have $\ef\in F$ and thus, $F'\subseteq F$. Since, conversely, $F\subseteq F'$ by the definition of $F'$, we have that $F=F'$ is an actual face of $\ES$.

For Statement~(\ref{item:A}), let $\bo\in M$. Then $A^{\ES}_\bo$ is an exposed face of $\ES$ containing $F$, since $\ef\cdot\bo=1$ for all $\ef\in F$. Then it follows from the actuality of $F$ that $F=A^{\ES}_\bo$.
\end{proof}

\maxmin*
\begin{proof}
By Lemma~\ref{lem:ESfaces} (see the note below the lemma), there exists $\tilde{\bo}\in\cS$ such that $F$ is the actual set of $\tilde{\bo}$ (in $\ES$), i.e.~$F=A^{\ES}_{\tilde{\bo}}=\{\ef\in\ES|\ef\cdot\tilde{\bo}=1\}$. Let $\ff\in\relint(F)$. Then we have $\ff\cdot\tilde{\bo}=1$ and $\ff\cdot\bo\leq1$ for all $\bo\in\cS$. Thus, $\ff$ defines a supporting hyperplane of $\cS$, intersecting at a non-empty exposed face $M=\{\bo\in\cS|\ff\cdot\bo=1\}\ni\tilde{\bo}$. It follows that $M$ satisfies the necessary and sufficient condition from the lemma as follows. If $\bo\in M$ then $\ff\cdot\bo=1$. Then since $\ff\in\relint{F}$ we have $\ef\cdot\bo=1$ for all $\ef\in F$. Conversely, if $\ef\cdot\bo=1$ for all $\ef\in F$, we have that $\ff\cdot\bo=1$ and thus, $\bo\in M$. 

We will now show that the exposed face $M$ is minimal. Suppose $M$ is not minimal. Then there exists an exposed face $M'\subset M$. By Lemma~\ref{lem:sfaces}, there exists $\ff'\in\ES$ such that $M'=\{\bo\in\cS|\ff'\cdot\bo=1\}$. Let $\boldsymbol{m}\in M
\setminus M'$. Then $\ff'\cdot\boldsymbol{m}<1$ which implies that $\ff'\notin F$ since $\boldsymbol{m}\in M$. On the other hand, let $\bo'\in M'$, and $F'=A^{\ES}_{\bo'}=\{\ef\in\ES|\ef\cdot\bo'=1\}$. Since $\bo'\in M$ we have that $F\subseteq F'$. However, we also have that $\ff'\in F'$. Thus, $F'$ is an exposed face of $\ES$ strictly containing $F$ and contradicting the maximality of $F$.
\end{proof}

\exposed*
\begin{proof}
Let $\cS$ be a convex set in which all extremal points are exposed. Consider a minimal exposed face $M$ containing a point $\bo$. Let $\bo=\sum_{j=1}^Np_j\bo_j$ be a convex decomposition of $\bo$ where $\bo_j$ are extremal points of $\cS$ for $1\leq j \leq N$. Any exposed face containing $\bo$ also contains the points $\bo_j$, however since $\bo_j$ are exposed and $M$ is minimal we find $\bo_j=\bo$ for all $1\leq j\leq N$ and $\bo$ is an exposed point of $\cS$. 
\end{proof}
\end{document}